\documentclass[sn-mathphys]{sn-jnl}% Math and Physical Sciences Reference Style
\jyear{2023}%

\theoremstyle{thmstyleone}%
\newtheorem{theorem}{Theorem}%  meant for continuous numbers
\newtheorem{lemma}[theorem]{Lemma}%
\newtheorem{corollary}[theorem]{Corollary}%

\theoremstyle{thmstyletwo}%

\theoremstyle{thmstylethree}%
\newtheorem{definition}{Definition}%

\usepackage{enumerate}

\newcommand{\seed}{\textsf{seed}}
\newcommand{\pk}{\textsf{pk}}

\DeclareMathOperator{\rank}{rank}
\newcommand{\transpose}{\intercal}

\begin{document}

    \title[Smaller public keys for MinRank-based schemes]{Smaller public keys for MinRank-based schemes%
        \footnote{The authors are members of GNSAGA of INdAM and of CrypTO, the group of Cryptography and Number Theory of the Politecnico di Torino.
            This work was partially supported by project SERICS (PE00000014) under the MUR National Recovery and Resilience Plan funded by the European Union -- NextGenerationEU.}}


    \author[1]{\fnm{Antonio J.} \sur{Di Scala}}\email{antonio.discala@polito.it}

    \author*[1]{\fnm{Carlo} \sur{Sanna}}\email{carlo.sanna@polito.it}
    \equalcont{These authors contributed equally to this work.}

    \affil*[1]{\orgdiv{Department of Mathematical Sciences}, \orgname{Politecnico di Torino}, \orgaddress{\street{Corso Duca degli Abruzzi 24}, \city{Torino}, \postcode{10129}, \country{Italy}}}


    \abstract{%
        MinRank is an NP-complete problem in linear algebra whose characteristics make it attractive to build post-quantum cryptographic primitives.
        Several MinRank-based digital signature schemes have been proposed.
        In particular, two of them, MIRA and MiRitH, have been submitted to the NIST Post-Quantum Cryptography Standardization Process.

        In this paper, we propose a key-generation algorithm for MinRank-based schemes that reduces the size of the public key to about 50\% of the size of the public key generated by the previous best (in terms of public-key size) algorithm.
        Precisely, the size of the public key generated by our algorithm sits in the range of 328--676 bits for security levels of 128--256 bits.
        We also prove that our algorithm is as secure as the previous ones.
    }



    \maketitle

    \section{Introduction}

    \emph{MinRank} is a problem in linear algebra that was first introduced by Buss, Frandsen, and Shallit (1999)~\cite{MR1705082}.
    Roughly speaking, given $k + 1$ matrices $M_0, \dots, M_k$ of size $m \times n$ over a finite field $\mathbb{F}_q$, the decisional version of MinRank asks to determine if there exists a non-trivial linear combination of $M_0, \dots, M_k$ whose rank does not exceed a fixed parameter $r$.
    The search version of MinRank, which is the one we will be focusing on hereafter, asks to find such a linear combination.

    For several reasons, MinRank is an attractive candidate to build post-quantum cryptographic primitives.
    First, MinRank is completely based on simple linear algebra operations, which can be implemented easily and efficiently.
    Second, the hardness of MinRank is supported by a long line of research: MinRank is an NP-complete problem~\cite{MR1705082} and, due to its relevance in cryptanalysis~\cite{MR4284269,MR3455911,MR4354587}, algorithms for solving it have been extensively studied, to the extent that random instances of MinRank are expected to be hard~\cite{Bardet2022107,cryptoeprint:2022/1031,MR4210317,MR3070108,MR2490380,MR2920562,MR1729291,MR3989004}.
    Finally, there are no known quantum algorithms to solve MinRank that go beyond straightforward quantum search applications.

    Several digital signature schemes based on MinRank have been proposed, namely: a scheme due to Courtois (2001)~\cite{MR1934855}, \textsf{MR-DSS} (2022)~\cite{Bellini2022144}, \textsf{MIRA}~(2023)~\cite{MIRA_preprint} (see also~\cite{cryptoeprint:2022/1512}), and \textsf{MiRitH}~(2023)~\cite{MiRitH_specs} (see also~\cite{ARZV22}).
    In particular, \textsf{MIRA} and \textsf{MiRitH} have been submitted to the NIST Post-Quantum Cryptography Standardization Process.

    In all these schemes, the public key is a random instance of MinRank, the secret key is the solution of such an instance, and the signing and verification algorithms together are a non-interactive zero-knowledge proof of knowledge of the solution.
    While the secret key can be easily compressed as a seed of $\lambda$ bits, where $\lambda$ is the security parameter, compressing the public key is less obvious.

    Courtois~\cite[Section~5.1]{MR1934855} proposed an algorithm, which we call \textsf{KeyGen1}, that compresses the public key in $\lambda + mn \log q$ bits, where $\log$ is the logarithm in base $2$.
    This method was improved in \textsf{MR-DSS}~\cite[Section~4.4]{Bellini2022144} by reducing the compressed public key to $\lambda + (mn - k) \log q$ bits.
    This improvement, which we call \textsf{KeyGen2}, is employed by \textsf{MIRA}~\cite[Section~2.4.1]{MIRA_preprint}, while \textsf{MiRitH} uses \textsf{KeyGen1}~\cite[Section~3.2]{MiRitH_specs}.

    We propose a new key-generation algorithm for MinRank-based schemes, which we call \textsf{KeyGen3}, with a compressed public key of $\lambda + (m(n - r) - k) \log q$ bits.
    (Note that $k < m(n - r)$.
    In fact, all parameter sets satisfy the stronger inequality $k < (m - r)(n - r)$, in order to make the MinRank problem \emph{overdetermined}, see~Section~\ref{sec:minrank}).

    Table~\ref{tab:comparison} provides a comparison of the sizes of the public keys\footnote{Hereafter, we will say ``public key'', respectively ``secret key'', instead of ``compressed public key'', respectively ``compressed secret key'', since the difference will be always clear from the context.} of the three key-generation algorithms, for the parameter sets proposed for \textsf{MiRitH}~\cite[Table~1]{MiRitH_specs}.
    As it can be seen, the public-key size of \textsf{KeyGen3} is about 50\% of that of \textsf{KeyGen2}, and sits in the range of 328--676 bits for security levels of 128--256 bits.

    \begin{table}[ht]
        \begin{center}
            \begin{tabular}{c@{\hskip 5pt}|c@{\hskip 10pt}c@{\hskip 10pt}c@{\hskip 10pt}c@{\hskip 10pt}c@{\hskip 10pt}|c@{\hskip 10pt}c@{\hskip 10pt}c@{\hskip 10pt}}
                \toprule
                & \multicolumn{5}{c|}{parameters} & \multicolumn{3}{c}{public key (bits)} \\
                $\lambda$ & $q$ & $m$ & $n$ & $k$ & $r$ & \textsf{KeyGen1} & \textsf{KeyGen2} & \textsf{KeyGen3}  \\
                \midrule
                $128$ & $16$ & $15$ & $15$ & $78$ & $6$ & $1,\!028$ & $716$ & $356$ \\
                $128$ & $16$ & $16$ & $16$ & $142$ & $4$ & $1,\!152$ & $584$ & $328$ \\
                $192$ & $16$ & $19$ & $19$ & $109$ & $8$ & $1,\!636$ & $1,\!200$ & $592$ \\
                $192$ & $16$ & $19$ & $19$ & $167$ & $6$ & $1,\!636$ & $968$ & $512$ \\
                $256$ & $16$ & $21$ & $21$ & $189$ & $7$ & $2,\!020$ & $1,\!264$ & $676$ \\
                $256$ & $16$ & $22$ & $22$ & $254$ & $6$ & $2,\!192$ & $1,\!176$ & $648$ \\
                \bottomrule
            \end{tabular}
            \vspace{5pt}
        \end{center}
        \caption{Comparison of the sizes of the public keys, for the parameter sets proposed for \textsf{MiRitH}~\cite[Table~1]{MiRitH_specs}.}\label{tab:comparison}
    \end{table}

    The next theorem reduces the security of \textsf{KeyGen3} to that of \textsf{KeyGen1}.
    For every $x > 0$, let $\tau(x) := \min(0.72, 2.1 x)$.

    \begin{theorem}\label{thm:main}
        Let $\mathcal{A}$ be an attacker that, given a random public key generated by \textsf{KeyGen1}, is able to efficiently retrieve the corresponding secret key with probability~$p_1$.
        If $\mathcal{A}$ is given a random public key generated by \textsf{KeyGen3}, then $\mathcal{A}$ can retrieve the corresponding secret key with probability $p_3 < \big(1 - \tau(q^{-1})\big)^{-4} p_1$.
    \end{theorem}

    Note that, if we take $q = 16$ as in Table~\ref{tab:comparison}, then $\big(1 - \tau(q^{-1})\big)^{-4} < 1.76$.
    The structure of the paper is the following.
    First, in Section~\ref{sec:prelim}, we provide the necessary notation (Section~\ref{sec:notation}), the formal definition of the MinRank problem (Section~\ref{sec:minrank}), and we recall the key-generation algorithm \textsf{KeyGen1} of Courtois (Section~\ref{sec:courtois}).
    Second, in Section~\ref{sec:new}, we describe our new key-generation algorithm \textsf{KeyGen3}.
    To simplify the exposition, we show first a partial (less efficient) version of the algorithm (Section~\ref{sec:first-impro}), and then, after recalling a canonical form for MinRank instances (Section~\ref{sec:canonical}), we show the complete algorithm (Section~\ref{sec:complete}).
    Finally, in Section~\ref{sec:proofs}, we prove Theorem~\ref{thm:main}.

    \section{Preliminaries}\label{sec:prelim}

    \subsection{Notation}\label{sec:notation}

    Let $\mathbb{F}_q$ be a finite field of $q$ elements.
    For all positive integers $m$, $n$, and $r \leq \min(m, n)$, let $\mathbb{F}_q^{m \times n}$ be the vector space of $m \times n$ matrices over $\mathbb{F}_q$, and let $\mathbb{F}_q^{m \times n, r}$ be the set of $m \times n$ matrices over $\mathbb{F}_q$ having rank equal to~$r$.
    For every $A \in \mathbb{F}_q^{m \times n}$, let $A^\transpose \in \mathbb{F}_q^{n \times m}$ be the transpose of $A$.
    Moreover, let $A^{\mathrm{L}} \in \mathbb{F}_q^{m \times (n - r)}$, respectively $A^{\mathrm{R}} \in \mathbb{F}_q^{m \times r}$, denote the matrix consisting of the first $n - r$, respectively the last $r$, columns of $A$, so that $A = (A^{\mathrm{L}} \mid A^{\mathrm{R}})$.
    Note that $r$ is omitted in the notation $A^{\mathrm{L}}$ and $A^{\mathrm{R}}$, but it will be always clear from the context.
    Let $\langle A \rangle \in \mathbb{F}_q^{1 \times mn}$ denote the row vector consisting of the entries of $A$ in column-major order, that is, the entries of $\langle A \rangle$ are, in order, the entries of the first column of $A$, followed by the entries of the second column of $A$, etc.
    Let $\langle A \rangle_i$ be the $i$th entry of $\langle A \rangle$.
    Let $I_s$, or just $I$ when the dimension is clear from the context, be the identity matrix of $\mathbb{F}_q^{s \times s}$.
    With a slight abuse of notation, let $0$ denote the zero matrix of $\mathbb{F}_q^{s \times t}$, the dimension $s \times t$ being always clear from the context.
    Finally, let $\delta_{i,j}$ be the Kronecker delta, let $\#S$ be the cardinality of the finite set $S$, and let $\lvert \textsf{obj} \rvert$ be the size in bits of the object $\textsf{obj}$.

    \subsection{MinRank}\label{sec:minrank}

    The search version of MinRank is formally defined as follows.

    \begin{definition}[MinRank]
        Let $q, m, n, k, r$ be positive integers, with $q$ a prime power and $m \geq n > r$.
        Given $k + 1$ matrices $M_0, \dots, M_k \in \mathbb{F}_q^{m \times n}$, the MinRank problem asks to find $\alpha_1, \dots, \alpha_k \in \mathbb{F}_q$ (if they exist) such that
        \begin{equation}\label{equ:E-def}
            E := M_0 + \sum_{i = 1}^k \alpha_i M_i
        \end{equation}
        has rank at most $r$.
    \end{definition}

    In MinRank-based schemes, the parameters $q,m,n,k,r$ are selected so that: every known algorithm to find a solution of MinRank with \mbox{$\rank(E) = r$} requires on average $2^\lambda$ operations; and random instances of MinRank are expected to have exactly one solution with overwhelming probability.
    Consequently, the schemes have to construct the solution so that $\rank(E) = r$.
    Furthermore, to enforce the uniqueness of the solution, it is required that MinRank is \emph{overdetermined}, that is, $k < (m - r)(n - r)$ \cite[p.~33]{MR3042659}.
    For details on the algorithms to solve MinRank, and consequentially on the selection of the parameters of MinRank-based schemes, see for example the documentation of \textsf{MiRitH}~\cite[Sections~4 and~5]{MiRitH_specs}.

    \subsection{The key-generation algorithm of Courtois}\label{sec:courtois}

    We begin by briefly reviewing the algorithms proposed by Courtois~\cite[Section~5.1]{MR1934855} to generate and decompress the public key and the secret key, see Figure~\ref{sch:courtois}.
    It is clear that $\textsf{KeyGen1}$ in Figure~\ref{sch:courtois} generates a random uniformly distributed instance of MinRank, and that the public key has a size of $\lvert \seed_{\pk} \rvert + \lvert M_0 \rvert = \lambda + mn \log q$ bits.
    The most computationally expensive step (not taking into account the cost of running the PRG) is the generation of $E$, which Courtois suggested to compute as $E = SLT$, where $L \in \mathbb{F}_q^{m \times n, r}$ is a fixed matrix and $S \in \mathbb{F}_q^{m \times m}$ and $T \in \mathbb{F}_q^{n \times n}$ are pseudorandom invertible matrices.
    %which in turn can be generated by their PLU decomposition.
    %In fact, one can directly generate $E$ by its PLU decomposition.

    \begin{figure}[ht]
        \begin{center}
            \includegraphics[width=\textwidth]{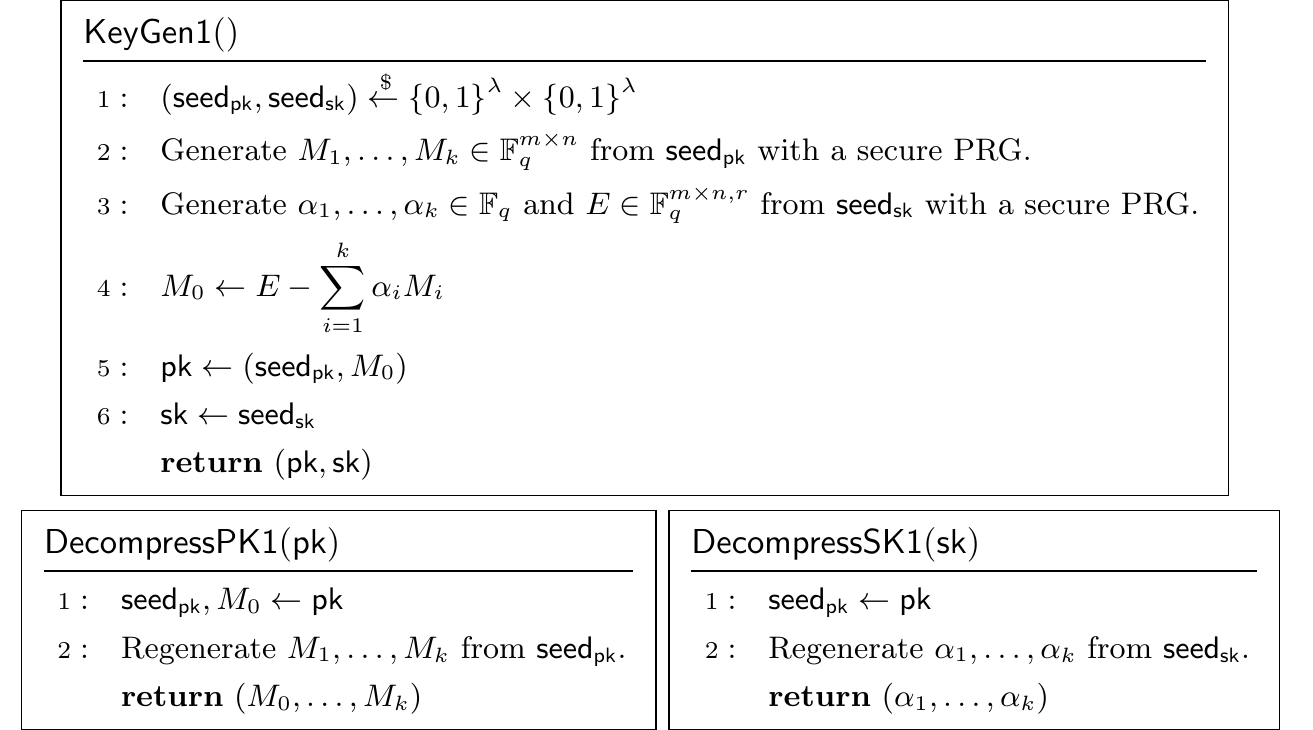}
        \end{center}
        \caption{The algorithms of Courtois to generate and decompress the keys.\protect\footnotemark}\label{sch:courtois}
    \end{figure}

    \section{New key-generation algorithm}\label{sec:new}

    \footnotetext{Actually, the key-generation algorithm in~\cite[Section~5.1]{MR1934855} is slightly different from that of Figure~\ref{sch:courtois} ($M_k$ plays the role of $M_0$, and consequently a division by $\alpha_k$ is necessary).
        However, this makes no difference in later arguments.
        We stated the key-generation algorithm this way only to uniformize it with the other algorithms.}

    \subsection{A first improvement}\label{sec:first-impro}

    To simplify the exposition, we provide first a key-generation algorithm with a public key of $\lambda + m(n - r) \log q$ bits.

    This algorithm employs the facts that: if $E \in \mathbb{F}_q^{m \times n, r}$ is taken at random with uniform distribution, then $E^{\mathrm{R}} \in \mathbb{F}_q^{m \times n, r}$ with significant probability (Lemma~\ref{lem:probability-EL-equal-ER-K}); and, in such a case, there exists a unique matrix $K \in \mathbb{F}_q^{r \times (n - r)}$ such that $E^{\mathrm{L}} = E^{\mathrm{R}} K$ (Lemma~\ref{lem:characterization-EL-equal-ER-K}).
    Then, assuming that $E^{\mathrm{L}} = E^{\mathrm{R}} K$, it follows from~\eqref{equ:E-def} that
    \begin{equation}\label{equ:M0L}
        M_0^{\mathrm{L}} = E^{\mathrm{R}} K -\sum_{i = 1}^k \alpha_i M_i^{\mathrm{L}} .
    \end{equation}
    Hence, we can generate pseudorandom $M_0^{\mathrm{R}}$, $M_1, \dots, M_k$, and $K$, compute
    \begin{equation}\label{equ:ER}
        E^{\mathrm{R}} = M_0^{\mathrm{R}} + \sum_{i=1}^k \alpha_i M_i^{\mathrm{R}}
    \end{equation}
    and $M_0^{\mathrm{L}}$ via~\eqref{equ:M0L}, and finally pack $M_0^{\mathrm{L}}$ into the public key.
    See Figure~\ref{sch:improved1} for the details.
    In this way, the size in bits of the public key is equal to
    \begin{equation*}
        \lvert \seed_{\pk} \rvert + \lvert M_0^{\textrm{L}} \rvert = \lambda + m(n - r) \log q .
    \end{equation*}
    Note that we cannot be sure that the matrix $E^{\mathrm{R}}$ computed by~\eqref{equ:ER} has full rank (this, by $E^{\mathrm{L}} = E^{\mathrm{R}} K$, is equivalent to $\rank(E) = r$).
    Therefore, we have to test if $\rank(E^{\mathrm{R}}) < r$ (step~5 of $\textsf{KeyGen}$ in Figure~\ref{sch:improved1}).
    Since $E^{\mathrm{R}}$ is a uniformly distributed random matrix in $\mathbb{F}_q^{m \times r}$, the probability that $E^{\mathrm{R}}$ is not full-rank is very small (less than $2^{-38.9}$ for the parameters in Table~\ref{tab:comparison}), see Lemma~\ref{lem:probability-full-rank}.
    Hence, the test has to be repeated only for a few times before finding a matrix $E^{\mathrm{R}}$ of full-rank.

    Furthermore, note that checking if $\rank(E^{\mathrm{R}}) < r$ must be done in way that prevents timing attacks, so either by a constant-time algorithm (see~\cite{10.1007/978-3-642-40349-1_15} for constant-time Gaussian elimination), or by a non-constant time algorithm that do not leak information about $E^{\mathrm{R}}$.
    For instance, one can multiply $E^{\mathrm{R}}$ on the left and on the right by random invertible matrices, and then check if the resulting product has rank less than $r$, so that the no information on $E^{\mathrm{R}}$ is leaked from the execution time.

    \begin{figure}[ht]
        \begin{center}
            \includegraphics[width=\textwidth]{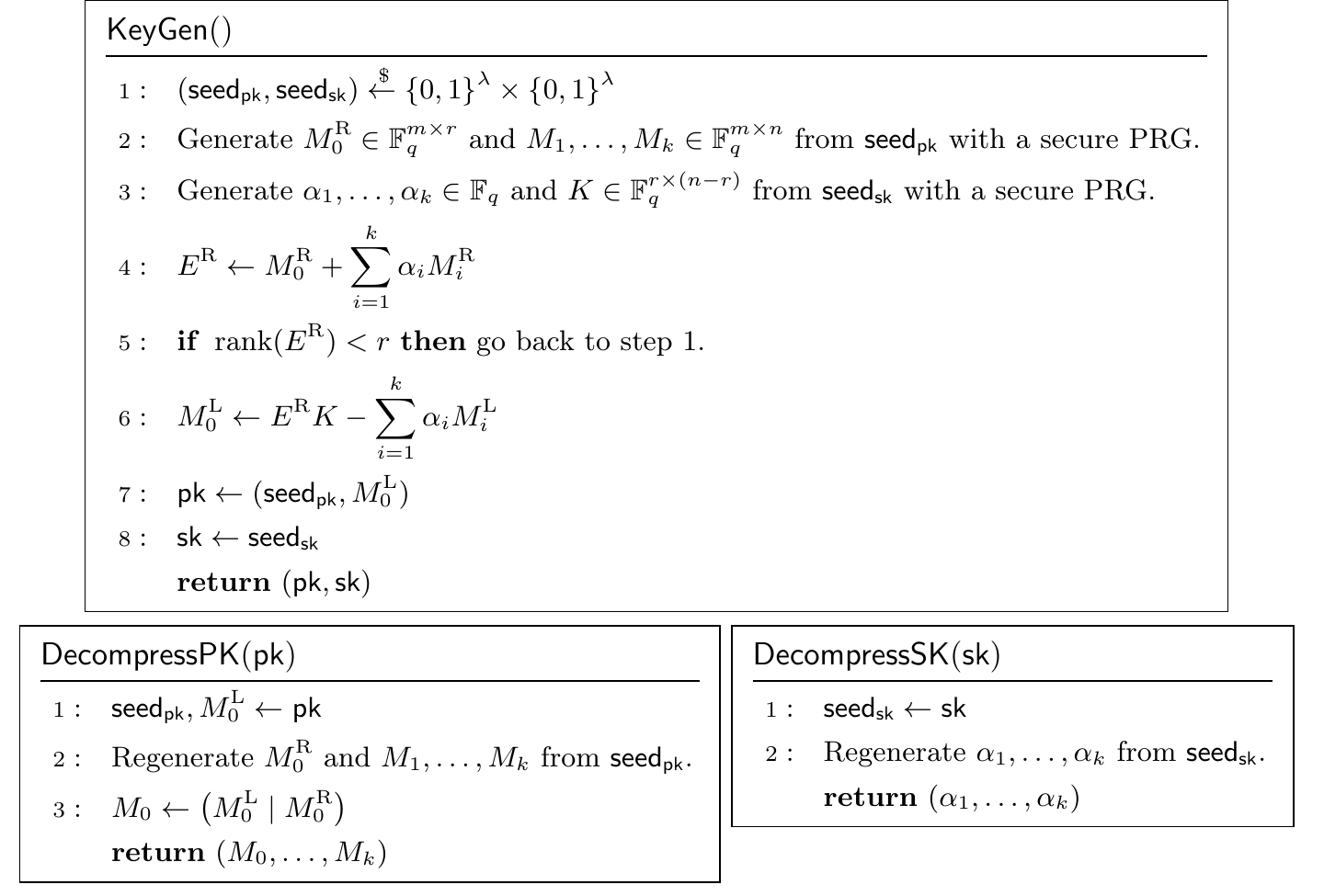}
        \end{center}
        \caption{First version of the improved key-generation algorithm.}\label{sch:improved1}
    \end{figure}

    \subsection{Canonical form of MinRank instances}\label{sec:canonical}

    In this section, we recall a canonical form of MinRank instances that was first introduced in~\cite[Section~4.4]{Bellini2022144}.

    Given a MinRank instance $\mathcal{M} = (M_0, \dots, M_k) \in (\mathbb{F}_q^{m \times n})^{k + 1}$, let $L \in \mathbb{F}_q^{(k + 1) \times mn}$ be the matrix whose rows are $\langle M_1 \rangle, \dots, \langle M_k \rangle$ and $\langle M_0 \rangle$, in this order.
    Furthermore, write
    \begin{equation*}
        L = \begin{pmatrix}\begin{array}{c|c}
                L_1 & L_2 \\[2pt] \hline \\[-9pt]
                \ell_1 & \ell_2
        \end{array}\end{pmatrix} ,
    \end{equation*}
    where $L_1 \in \mathbb{F}_q^{k \times k}$, $L_2 \in \mathbb{F}_q^{k \times (mn - k)}$, $\ell_1 \in \mathbb{F}_q^{1 \times k}$, and $\ell_2 \in \mathbb{F}_q^{1 \times (mn - k)}$.

    If $L_1$ is invertible, then we say that $\mathcal{M}$ is \emph{reducible to canonical form} and that the \emph{canonical form} of $\mathcal{M}$ is $\mathcal{M}^\prime := (M_0^\prime, \dots, M_k^\prime) \in (\mathbb{F}_q^{m \times n})^{k + 1}$, where $\langle M_1^\prime \rangle, \dots, \langle M_k^\prime \rangle$ and $\langle M_0^\prime \rangle$ are the rows, in this order, of the matrix
    \begin{equation*}
        L^\prime := \begin{pmatrix}\begin{array}{c|c}
                L_1^{-1} & 0 \\[2pt] \hline \\[-9pt]
                -\ell_1 L_1^{-1} & 1
        \end{array}\end{pmatrix} L = \begin{pmatrix}\begin{array}{c|c}
                I_k & L_1^{-1} L_2 \\[2pt] \hline \\[-9pt]
                0 & \ell_2 - \ell_1 L_1^{-1} L_2
        \end{array}\end{pmatrix} .
    \end{equation*}
    In particular, we have that $(M_0^\prime, \dots, M_k^\prime) \in \mathcal{C}_0 \times \mathcal{C}_1$, where
    \begin{equation*}
        \mathcal{C}_0 := \big\{N \in \mathbb{F}_q^{m \times n} : \langle N \rangle_i = 0 \text{ for } i \in \{1, \dots, k\}\big\}
    \end{equation*}
    and
    \begin{equation*}
        \mathcal{C}_1 := \big\{(N_1, \dots, N_k) \in \mathbb{F}_q^{m \times n} : \langle N_i \rangle_j = \delta_{i,j} \text{ for } i, j \in \{1, \dots, k\}\big\} .
    \end{equation*}
    In general, we say that MinRank instances belonging to $\mathcal{C}_0 \times \mathcal{C}_1$ are in \emph{canonical form}.
    If $\mathcal{M}$ is reducible to the canonical form $\mathcal{M}^\prime$, then an easy computation shows that \eqref{equ:E-def} is equivalent to
    \begin{equation*}
        E := M_0^\prime + \sum_{i = 1}^k \alpha_i^\prime M_i^\prime ,
    \end{equation*}
    where
    \begin{equation}\label{equ:canonical-alpha}
        \begin{pmatrix} \alpha_1^\prime & \cdots & \alpha_k^\prime \end{pmatrix}
        = \begin{pmatrix} \alpha_1 & \cdots & \alpha_k \end{pmatrix} L_1 + \ell_1 .
    \end{equation}
    Consequently, finding a solution to the instance $\mathcal{M}$ is equivalent to finding a solution to the instance $\mathcal{M}^\prime$.

    \subsection{The complete algorithm}\label{sec:complete}

    \begin{figure}[ht]
        \begin{center}
            \includegraphics[width=\textwidth]{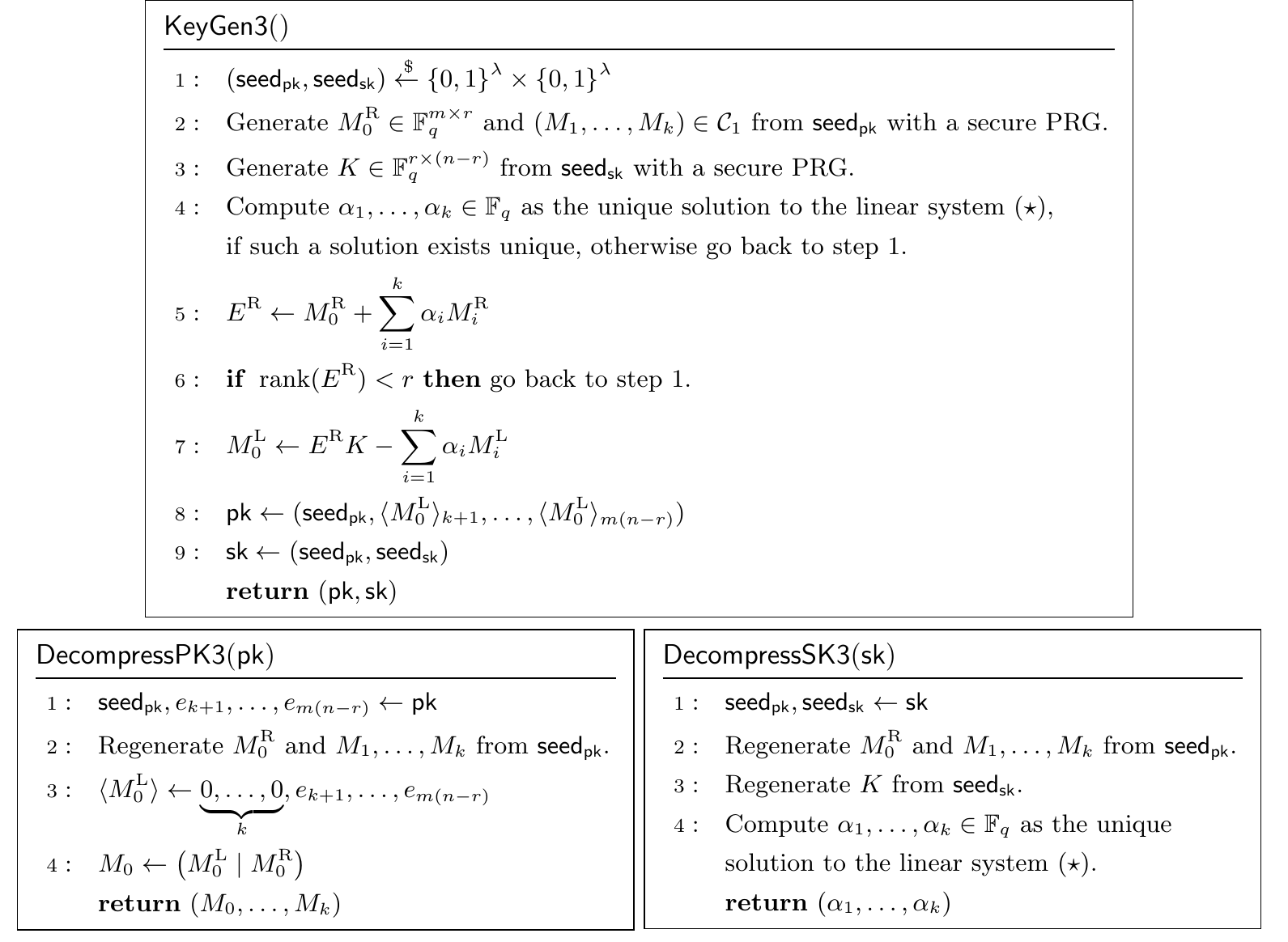}
        \end{center}
        \caption{The proposed key-generation algorithm.}\label{sch:improved2}
    \end{figure}

    Now we can provide the key-generation algorithm with a public key of $\lambda + (m(n - r) - k) \log q$ bits.

    The idea is to generate $M_0, \dots, M_k$ so that they are in canonical form.
    In this way, the first $k$ entries of $\langle M_0^{\mathrm{L}} \rangle$ are equal to $0$, and there is no need to pack them into the public key.
    Thus the size of the public key is reduced to $\lambda + (m(n - r) - k) \log q$ bits.

    The \textsf{KeyGen} algorithm of Figure~\ref{sch:improved1} can be easily modified to generate $(M_1, \dots, M_k) \in \mathcal{C}_1$.
    However, the way in which $M_0^{\mathrm{L}}$ is computed does not guarantee that $M_0, \dots, M_k$ are in canonical form, i.e., that $M_0 \in \mathcal{C}_0$.
    To achieve that, we have to choose $\alpha_1, \dots, \alpha_k$ so that the first $k$ entries of $\langle M_0^{\mathrm{L}} \rangle$ are equal to $0$.
    Since
    \begin{equation*}
        M_0^{\mathrm{L}} = \left(M_0^{\mathrm{R}} + \sum_{j = 1}^k \alpha_j M_j^{\mathrm{R}}\right) \! K - \sum_{j = 1}^k \alpha_j M_j^{\mathrm{L}}
    \end{equation*}
    and $\langle M_i^{\mathrm{L}} \rangle_j = \delta_{i, j}$ for $i, j \in \{1, \dots, k\}$ (note that $k < m(n - r)$), this amount to solving the linear system
    \begin{equation}\label{equ:alpha-system}\tag{$\star$}
        \sum_{j = 1}^k \left(\delta_{i,j} - \langle M_j^{\mathrm{R}} K \rangle_i \right) \alpha_j = \langle M_0^{\mathrm{R}} K \rangle_i \quad (i=1,\dots,k) .
    \end{equation}
    We will prove that~\eqref{equ:alpha-system} has a unique solution with high probability (Lemma~\ref{lem:probability-I-X}).
    The algorithms for the generation of the keys and their decompression are given in Figure~\ref{sch:improved2}.

    Note that solving~\eqref{equ:alpha-system} must be done in constant time, in order to protect the secret $\alpha_1, \dots, \alpha_k$ from timing attacks.
    Furthermore, note that this construction requires to store $\seed_{\pk}$ into the secret key.
    However, this should not be an issue since, usually, whoever has the secret key also has the public key.

    \section{Proof of Theorem~\ref{thm:main}}\label{sec:proofs}

    \subsection{Preliminaries}

    In this section, we collect some preliminary lemmas.
    We begin with the following inequality.

    \begin{lemma}\label{lem:tau-bound}
        We have that
        \begin{equation}\label{equ:q-product}
            \prod_{j = s}^\infty (1 - q^{-j}) > 1 - \tau(q^{-s})
        \end{equation}
        for all integers $s \geq 1$.
    \end{lemma}
    \begin{proof}
        Let $P_s(q)$ denote the product in~\eqref{equ:q-product}.
        First, suppose that $q^{s + 1} \geq 8$.
        Since the logarithm is concave, we have that $\ln(1 - x) \geq - c_0 x$, for all $x_0 \in {(0, 1)}$ and $x \in [0, x_0]$, where
        \begin{equation*}
            c_0 = c_0(x_0) := -\frac{\ln(1 - x_0)}{x_0} > 0.
        \end{equation*}
        Hence, taking $x_0 = q^{-(s+1)}$, we get that
        \begin{align*}
            P_{s+1}(q) \geq \exp\!\left(-c_0 \sum_{j = s + 1}^\infty q^{-j}\right) = \exp\!\left(-\frac{c_0 q^{-(s+1)}}{1 - q^{-1}}\right) > 1 - \frac{c_0 q^{-(s+1)}}{1 - q^{-1}} ,
        \end{align*}
        where we also used the fact that $\exp(-x) > 1 - x$ for all $x > 0$.
        Therefore, we obtain that
        \begin{equation*}
            P_s(q) > (1 - q^{-s})\left(1 - \frac{c_0 q^{-(s+1)}}{1 - q^{-1}}\right) > 1 - \left(1 + \frac{c_0}{q - 1}\right) q^{-s} .
        \end{equation*}
        Since $c_0(x_0)$ is an increasing function of $x_0$, it follows that $c_0(x_0) \leq c_0(1/8) < 1.1$.
        Hence, we get that
        \begin{equation*}
            P_s(q) > 1 - (1 + c_0) q^{-s} > 1 - 2.1 q^{-s} = 1 - \tau(q^{-s}) ,
        \end{equation*}
        since $2.1 q^{-s} < 0.72$.

        Now suppose that $q^{s + 1} < 8$.
        Then $q = 2$ and $s = 1$.
        Moreover, we get that
        \begin{align*}
            P_s(q) &= (1 - 2^{-1})(1 - 2^{-2})(1 - 2^{-3}) P_4(2) \\
            &> (1 - 2^{-1})(1 - 2^{-2})(1 - 2^{-3}) (1 - 2.1 \cdot 2^{-4}) \\
            &> 1 - 0.72 \\
            &= 1 - \tau(q^{-s}) ,
        \end{align*}
        since $2.1 q^{-s} > 0.72$.
        The proof is complete.
    \end{proof}

    The next lemma provides a formula for the number of $m \times n$ matrices of rank $r$ over $\mathbb{F}_q$.

    \begin{lemma}\label{lem:rank-count}
        We have that
        \begin{equation*}
            \# \mathbb{F}_q^{m \times n, r} = \prod_{i = 0}^{r - 1} \frac{(q^m - q^i)(q^n - q^i)}{q^r - q^i} .
        \end{equation*}
    \end{lemma}
    \begin{proof}
        See, e.g., \cite{MR1533848}.
    \end{proof}

    The next three results are well known (more or less in these forms), but we include their proofs for completeness.

    \begin{lemma}\label{lem:probability-full-rank}
        Let $s, t$ be positive integers, and let $A \in \mathbb{F}_q^{s \times t}$ be a random matrix taken with uniform distribution.
        Then the probability that $\rank(A) = \min(s,t)$ is greater than $1 - \tau(q^{-\lvert s - t \rvert - 1})$.
    \end{lemma}
    \begin{proof}
        Since $\rank(A^\transpose) = \rank(A)$, we can assume that $s \geq t$.
        Hence, the probability that $\rank(A) = \min(s, t)$ is equal to the probability that $A \in \mathbb{F}_q^{s \times t, t}$.
        In turn, by Lemma~\ref{lem:rank-count}, such a probability is equal to
        \begin{equation*}
            \frac{\#\mathbb{F}_q^{s \times t, t}}{\#\mathbb{F}_q^{s \times t}} = \left(\prod_{i=0}^{t - 1} (q^s - q^i)\right) \cdot q^{-st} = \prod_{i=0}^{t - 1} (1 - q^{i-s}) > \prod_{j = s - t + 1}^\infty (1 - q^{-j}) ,
        \end{equation*}
        and the claim follows from Lemma~\ref{lem:tau-bound}.
    \end{proof}

    \begin{corollary}\label{cor:probability-invertible}
        Let $s$ be a positive integer and let $A \in \mathbb{F}_q^{s \times s}$ be a random matrix taken with uniform probability.
        Then the probability that $A$ is invertible is greater than $1 - \tau(q^{-1})$.
    \end{corollary}

    \begin{lemma}\label{lem:random-product}
        Let $A \in \mathbb{F}_q^{s \times s, s}$ be a random matrix with an arbitrary probability distribution, and let $B \in \mathbb{F}_q^{s \times t}$ (respectively $C \in \mathbb{F}_q^{t \times s}$) be a random uniformly distributed matrix independent from $A$.
        Then the matrix $AB$ (respectively $CA$) is uniformly distributed in $\mathbb{F}_q^{s \times t}$ (respectively $\mathbb{F}_q^{t \times s}$).
    \end{lemma}
    \begin{proof}
        It suffices to prove the claim for $B$.
        Then, the claim for $C$ follows by matrix transposition.
        For each $D \in \mathbb{F}_q^{s \times t}$, we have that
        \begin{align*}
            \Pr[AB = D] &= \sum_{A_0 \in \mathbb{F}_q^{s \times s, s}} \Pr[A = A_0] \Pr[B = A_0^{-1} D] \\
            &= \sum_{A_0 \in \mathbb{F}_q^{s \times s, s}} \Pr[A = A_0] \, \frac1{\# \mathbb{F}_q^{s \times t}} = \frac1{\# \mathbb{F}_q^{s \times t}} .
        \end{align*}
        Hence, we get that $AB$ is uniformly distributed in $\mathbb{F}_q^{s \times t}$.
    \end{proof}

    Let $\mathcal{E}$ be the set of $E \in \mathbb{F}_q^{m \times n, r}$ such that $E^{\mathrm{R}} \in \mathbb{F}_q^{m \times r, r}$.

    \begin{lemma}\label{lem:characterization-EL-equal-ER-K}
        Let $E \in \mathbb{F}_q^{m \times n, r}$.
        Then $E \in \mathcal{E}$ if and only if $E^{\mathrm{L}} = E^{\mathrm{R}} K$ for some $K \in \mathbb{F}_q^{r \times (n - r)}$.
        In such a case, we have that $K$ is unique.
    \end{lemma}
    \begin{proof}
        First, suppose that $E \in \mathcal{E}$.
        Then the columns of $E^{\mathrm{R}}$ generate the column-space of $E$.
        Consequently, the columns of $E^{\mathrm{L}}$ are a linear combination of those of $E^{\mathrm{R}}$, that is, $E^{\mathrm{L}} = E^{\mathrm{R}} K$ for some $K \in \mathbb{F}_q^{r \times (n - r)}$.
        Moreover, the matrix $K$ is unique, since the columns of $E^{\mathrm{R}}$ are linearly independent.
        Vice versa, if $E^{\mathrm{L}} = E^{\mathrm{R}} K$ for some $K \in \mathbb{F}_q^{r \times (n - r)}$, then the column-space of $E$ is generated by the columns of $E^{\mathrm{R}}$.
        Since $E$ has rank $r$, it follows that $E^{\mathrm{R}} \in \mathbb{F}_q^{m \times r, r}$, that is $E \in \mathcal{E}$.
    \end{proof}

    \begin{lemma}\label{lem:probability-EL-equal-ER-K}
        Let $E \in \mathbb{F}_q^{m \times n, r}$ be a random matrix taken with uniform distribution.
        Then $E \in \mathcal{E}$ with probability greater that $1 - \tau(q^{-1})$.
        In such a case, the unique matrix $K \in \mathbb{F}_q^{r \times (n - r)}$ such that $E^{\mathrm{L}} = E^{\mathrm{R}} K$ (see Lemma~\ref{lem:characterization-EL-equal-ER-K}) is uniformly distributed in $\mathbb{F}_q^{r \times (n - r)}$.
    \end{lemma}
    \begin{proof}
        By Lemma~\ref{lem:characterization-EL-equal-ER-K}, the map $\Phi$ that sends each $E \in \mathcal{E}$ to $(E^{\mathrm{R}}, K)$, where $K \in \mathbb{F}_q^{r \times (n - r)}$ is the unique matrix such that $E^{\mathrm{L}} = E^{\mathrm{R}} K$, is a bijection
        \begin{equation*}
            \mathcal{E} \to \mathbb{F}_q^{m \times r, r} \times \mathbb{F}_q^{r \times (n - r)} .
        \end{equation*}
        Hence, by Lemma~\ref{lem:rank-count}, the probability that $E \in \mathcal{E}$ is equal to
        \begin{align*}
            \frac{\#\mathbb{F}_q^{m \times r, r} \cdot \#\mathbb{F}_q^{r \times (n - r)}}{\#\mathbb{F}_q^{m \times n, r}}
            &= \left(\prod_{i = 0}^{r - 1} (q^m - q^i) \right) \cdot q^{r(n - r)} \cdot \left(\prod_{i = 0}^{r - 1} \frac{(q^m - q^i)(q^n - q^i)}{q^r - q^i}\right)^{-1} \\
            &\hspace{-7em}= \prod_{i = 0}^{r - 1} \frac{(q^r - q^i) q^{n - r}}{q^n - q^i} = \prod_{i = 0}^{r - 1} \frac{1 - q^{i - r}}{1 - q^{i - n}} > \prod_{i = 0}^{r - 1} (1 - q^{i - r}) \\
            &\hspace{-7em}> \prod_{j=1}^\infty (1 - q^{-j}) > 1 - \tau(q^{-1}) ,
        \end{align*}
        where the last inequality follows from Lemma~\ref{lem:tau-bound}.

        Furthermore, again since $\Phi$ is a bijection, we get that $K$ is uniformly distributed in $\mathbb{F}_q^{r \times (n - r)}$.
    \end{proof}

    The next lemma regards the probability that a MinRank instance can be reduced to canonical form, and the distributions of its canonical form and the corresponding solution.

    \begin{lemma}\label{lem:probability-canonical-form}
        Let $M_0, M_1, \dots, M_k$ be a random MinRank instance and let $\alpha_1, \dots, \alpha_k$ be the corresponding solution.
        Assume that $M_1, \dots, M_k$ and $\alpha_1, \dots, \alpha_k$ are independent and uniformly distributed in $\mathbb{F}_q^{m \times r}$ and $\mathbb{F}_q$, respectively (while $M_0$ depends on $M_1, \dots, M_k$).
        Then $M_0, \dots, M_k$ can be reduced to canonical form with probability greater than $1 - \tau(q^{-1})$.
        In such a case, letting $M_0^{\prime}, \dots, M_k^{\prime}$ be the canonical form of $M_0, \dots, M_k$, and letting $\alpha_1^\prime, \dots, \alpha_k^\prime$ be given by~\eqref{equ:canonical-alpha}, we have that $(M_1^{\prime}, \dots, M_k^\prime)$ and $(\alpha_1^\prime, \dots, \alpha_k^\prime)$ are independent and uniformly distributed in $\mathcal{C}_1$ and $\mathbb{F}_q^k$, respectively.
    \end{lemma}
    \begin{proof}
        With the notation of Section~\ref{sec:canonical}, we have that
        \begin{equation*}
            \begin{pmatrix} \langle M_1 \rangle \\ \vdots \\ \langle M_k \rangle \end{pmatrix} = \begin{pmatrix}\!\begin{array}{c|c} L_1 & L_2 \end{array}\!\end{pmatrix} .
        \end{equation*}
        Hence, it follows that $L_1 \in \mathbb{F}_q^{k \times k}$ and $L_2 \in \mathbb{F}_q^{k \times (mn - k)}$ are independent and uniformly distributed.
        Since $M_0, \dots, M_k$ can be reduced to canonical form exactly when the matrix $L_1$ is invertible, it follows from Corollary~\ref{cor:probability-invertible} that the probability that the reduction is possible is greater than $1 - \tau(q^{-1})$.
        Furthermore, if $L_1$ is invertible, we have that
        \begin{equation*}
            \begin{pmatrix} \langle M_1^\prime \rangle \\ \vdots \\ \langle M_k^\prime \rangle \end{pmatrix} = \begin{pmatrix}\!\begin{array}{c|c} I_k & L_1^{-1} L_2 \end{array}\!\end{pmatrix} ,
        \end{equation*}
        and the claim about the distribution of $(M_1^{\prime}, \dots, M_k^\prime)$ and $(\alpha_1^\prime, \dots, \alpha_k^\prime)$ follows from Lemma~\ref{lem:random-product}.
    \end{proof}

    We conclude with a lemma concerning the invertibility of a certain matrix.

    \begin{lemma}\label{lem:probability-I-X}
        Let $N_1, \dots, N_k \in \mathbb{F}_q^{m \times r}$ and $K \in \mathbb{F}_q^{r \times (n - r)}$ be random matrices that are independent and uniformly distributed in their respective spaces.
        Let $X \in \mathbb{F}_q^{k \times k}$ be the matrix whose entry of the $i$th row and $j$th column is equal to $\langle N_j K \rangle_i$.
        Then
        \begin{equation*}
            \Pr\!\left[I - X \in \mathbb{F}_q^{k \times k, k}\right] > \big(1 - \tau(q^{-1})\big)^2 .% \rho\big(\!\min(mr, k)\big) \, \rho(r) .
        \end{equation*}
    \end{lemma}
    \begin{proof}
        Let $\rho(s)$ be the probability that a uniformly distributed random matrix in $\mathbb{F}_q^{s \times s}$ is invertible.
        Write $K = (K_1 \mid K_2)$, where $K_1 \in \mathbb{F}_q^{r \times r}$ and $K_2 \in \mathbb{F}_q^{r \times (n - 2r)}$.
        Note that
        \begin{align}\label{equ:conditional-probability}
            \Pr&\!\left[I - X \in \mathbb{F}_q^{k \times k, k}\right]
            \geq \Pr\!\left[I - X \in \mathbb{F}_q^{k \times k, k} \text{ and } K_1 \in \mathbb{F}_q^{r \times r, r}\right] \nonumber\\
            &= \Pr\!\left[I - X \in \mathbb{F}_q^{k \times k, k} \mid K_1 \in \mathbb{F}_q^{r \times r, r}\right] \, \Pr\!\left[K_1 \in \mathbb{F}_q^{r \times r, r}\right] \nonumber\\
            &= \Pr\!\left[I - X \in \mathbb{F}_q^{k \times k, k} \mid K_1 \in \mathbb{F}_q^{r \times r, r}\right] \, \rho(r) .
        \end{align}
        Therefore, it suffices to prove that the conditional probability in~\eqref{equ:conditional-probability} is equal to $\rho\big(\!\min(mr, k)\big)$, and then the claim follows from Corollary~\ref{cor:probability-invertible}.

        Hereafter, assume that $K_1$ is invertible.
        Let $N_j^\prime := N_j K_1$ for each $j \in \{1, \dots, k\}$.
        By Lemma~\ref{lem:random-product}, we have that $N_1^\prime, \dots, N_k^\prime$ are independent and uniformly distributed in $\mathbb{F}_q^{m \times r}$.
        Moreover, we have that $N_j K = (N_j^\prime \mid N_j^\prime K_1^{-1} K_2)$ for each $j \in \{1, \dots, k\}$.
        Consequently, we get that $\langle N_j K \rangle_i = \langle N_j^\prime \rangle_i$ for all positive integers $i \leq \min(mr, k)$.

        If $mr \geq k$, then it follows that $\langle N_j K \rangle = \langle N_j^\prime \rangle$ for each $i \in \{1, \dots, k\}$.
        Hence, $X$ is uniformly distributed in $\mathbb{F}_q^{k \times k}$.
        Thus the conditional probability in~\eqref{equ:conditional-probability} is equal to $\rho(k)$, as desired.

        Assume that $mr < k$.
        It follows easily that there exists a matrix $H \in \mathbb{F}_q^{mr \times (k - mr)}$, which is completely determined by $K$, such that $X = (I_{mr} \mid H)^\transpose \, J$, where
        \begin{equation*}
            J := \begin{pmatrix} \langle N_1^\prime \rangle^\transpose & \cdots & \langle N_k^\prime \rangle^\transpose \end{pmatrix}
        \end{equation*}
        is uniformly distributed in $\mathbb{F}_q^{mr \times k}$.

        Note that the matrix
        \begin{equation*}
            P := \begin{pmatrix}\begin{array}{c|c}
                    I_{mr} & 0 \\[5pt] \hline \\[-8pt]
                    H^\transpose & -I_{k - mr}
            \end{array}\end{pmatrix}
        \end{equation*}
        satisfies $P (I_{mr} \mid H)^\transpose = (I_{mr} \mid 0)^\transpose$ and $P^2 = I$.
        In particular, $P$ is invertible.
        Hence, by Lemma~\ref{lem:random-product}, we have that $J^\prime := J P$ is uniformly distributed in $\mathbb{F}_q^{mr \times k}$.
        Write $J^\prime = (J_1^\prime \mid J_2^\prime)$, where $J_1^\prime \in \mathbb{F}_q^{mr \times mr}$ and $J_2^\prime \in \mathbb{F}_q^{mr \times (k - mr)}$ are independent and uniformly distributed.
        Then we have that
        \begin{align*}
            P(I - X)P &= P^2 - P X P = I - P (I_{mr} \mid H)^\transpose \, J P \\
            &= I - (I_{mr} \mid 0)^\transpose \, J^\prime
            = \begin{pmatrix}\begin{array}{c|c}
                    I_{mr} - J_1^\prime \, & \, -J_2^\prime \\[2pt] \hline \\[-8pt]
                    0 \, & \, I_{k - mr}
            \end{array}\end{pmatrix} .
        \end{align*}
        Consequently, we get that $I - X$ is invertible if and only if $I - J_1^\prime$ is invertible.
        Therefore, the conditional probability in~\eqref{equ:conditional-probability} is equal to $\rho(mr)$, as desired.
    \end{proof}

    \subsection{Proof of Theorem~\ref{thm:main}}

    Our strategy to prove Theorem~\ref{thm:main} is the following.
    First, we provide an algorithm $\mathcal{R}$ that, starting from a random instance of MinRank $\mathcal{M}$ generated by \textsf{KeyGen1}, returns the canonical form $\mathcal{M}^\prime$ of $\mathcal{M}$ with probability greater than $\big(1 - \tau(q^{-1})\big)^4$.
    Then, we show that $\mathcal{M}^\prime$ follows the same probability distribution of a random MinRank instance generated by \textsf{KeyGen3}.
    Therefore, since the attacker $\mathcal{A}$ can solve $\mathcal{M}^\prime$ (and thus $\mathcal{M}$) with probability $p_3$, we get that the attacker $\mathcal{A}$ can solve $\mathcal{M}$ with probability $p_1 > \big(1 - \tau(q^{-1})\big)^4 p_3$.
    Consequently, it follows that $p_3 < \big(1 - \tau(q^{-1})\big)^{-4} p_1$, as desired.

    The steps of the algorithm $\mathcal{R}$ are the following.

    \begin{enumerate}
        \vspace{0.5em}
        \setlength{\itemsep}{0.5em}

        \item\label{ite:R1} Generate $M_0, \dots, M_k \in \mathbb{F}_q^{m \times n}$, $E \in \mathbb{F}_q^{m \times n, r}$, and $\alpha_1, \dots, \alpha_k \in \mathbb{F}_q$ as they are generated by \textsf{KeyGen1}.
        In particular, we have that $M_1, \dots, M_k$, $E$, and $\alpha_1, \dots, \alpha_k$ are independent and uniformly distributed in their respective spaces.

        \item If $E^{\mathrm{R}} \notin \mathbb{F}_q^{m \times r, r}$ then stop.
        Otherwise, if $E^{\mathrm{R}} \in \mathbb{F}_q^{m \times r, r}$, then, in light of Lemma~\ref{lem:characterization-EL-equal-ER-K}, let $K \in \mathbb{F}_q^{r \times (n - r)}$ be the unique matrix such that $E^{\mathrm{L}} = E^{\mathrm{R}} K$.
        Note that, by Lemma~\ref{lem:probability-EL-equal-ER-K},
        the second case happens with probability greater than $1 - \tau(q^{-1})$,
        and $K$ is uniformly distributed in $\mathbb{F}_q^{r \times (n - r)}$.

        \item\label{ite:R3} If $M_0, \dots, M_k$ cannot be reduced to canonical form then stop.
        Otherwise, if $M_0, \dots, M_k$ can be reduced to canonical form, let $M_0^\prime, \dots, M_k^\prime$ be the canonical form of $M_0, \dots, M_k$, and let $\alpha_1^\prime, \dots, \alpha_k^\prime$ be given by~\eqref{equ:canonical-alpha}.
        Note that, by Lemma~\ref{lem:probability-canonical-form}, the second case happens with probability greater than $1 - \tau(q^{-1})$, while
        $(M_1^\prime, \dots, M_k^\prime)$ and $(\alpha_1^\prime, \dots, \alpha_k^\prime)$ are independent and uniformly distributed in $\mathcal{C}_1$ and $\mathbb{F}_q^k$, respectively.
        Furthermore, since $k < m(n - r)$, we have that $M_1^{\prime\mathrm{R}}, \dots, M_k^{\prime\mathrm{R}}$ are independent and uniformly distributed in $\mathbb{F}_q^{m \times r}$.

        \item Let $X \in \mathbb{F}_q^{k \times k}$ be the matrix having the entry of the $i$th row and $j$th column equal to $\langle M_j^{\prime\mathrm{R}} K \rangle_i$.
        If $I - X$ is not invertible then stop.
        Otherwise, return $M_0^\prime, \dots, M_k^\prime$.
        Note that, by Lemma~\ref{lem:probability-I-X}, the second case happens with probability greater than $\big(1 - \tau(q^{-1})\big)^2$.

        \vspace{0.5em}
    \end{enumerate}

    By construction, we have that $\mathcal{R}$ returns the canonical form $M_0^\prime, \dots, M_k^\prime$ with probability greater than $\big(1 - \tau(q^{-1})\big)^4$.
    It remains to prove that such canonical form follows the same probability distribution of a MinRank instance generated by \textsf{KeyGen3}.

    Let $\mathcal{S}$ be the set of
    \begin{equation*}
        (M_0^*, \dots, M_k^*, E^*, \alpha_1^*, \dots, \alpha_k^*) \in \mathcal{C}_0 \times \mathcal{C}_1 \times \mathcal{E} \times \mathbb{F}_q^k
    \end{equation*}
    such that
    \begin{enumerate}[(i)]
        \vspace{0.5em}
        \setlength{\itemsep}{0.5em}

        \item\label{ite:c1} $E^* = M_0^* + \sum_{i = 1}^k \alpha_i^* M_i^*$;

        \item\label{ite:c2} $\alpha_1^*, \dots, \alpha_k^*$ is the unique solution to the linear system
        \begin{equation*}
            \sum_{j = 1}^k \left(\delta_{i,j} - \langle M_j^{*\mathrm{R}} K\rangle_i \right) x_i = \langle M_0^{*\mathrm{R}} K^* \rangle_i \quad (i = 1, \dots, k) ,
        \end{equation*}
        where $K^* \in \mathbb{F}_q^{r \times (n - r)}$ is the unique matrix such that $E^{*\mathrm{L}} = E^{*\mathrm{R}} K^*$, by Lemma~\ref{lem:characterization-EL-equal-ER-K}.

        \vspace{0.5em}
    \end{enumerate}
    Note that each element of $\mathcal{S}$ is completely determined by either
    \begin{enumerate}[(a)]
        \vspace{0.5em}
        \setlength{\itemsep}{0.5em}

        \item\label{ite:n1} $M_1^*, \dots, M_k^*$, $E^*$, and $\alpha_1^*, \dots, \alpha_k^*$, since using~\eqref{ite:c1} one can retrieve $M_0^*$ from the former matrices and scalars; or

        \item\label{ite:n2} $M_0^{*\mathrm{R}}, M_1^*, \dots, M_k^*$, and $K^*$.
        In fact, given such matrices, one can retrieve $\alpha_1^*, \dots, \alpha_k^*$ by using~\eqref{ite:c2}.
        Then, using~\eqref{ite:c1}, one gets that
        \begin{equation*}
            E^{* \mathrm{R}} = M_0^{* \mathrm{R}} + \sum_{i = 1}^k \alpha_i^* M_i^{* \mathrm{R}} .
        \end{equation*}
        Finally, one has that $E^* = \big(E^{*\mathrm{R}} K^* \mid E^{*\mathrm{R}}\big)$.

        \vspace{0.5em}
    \end{enumerate}
    Assume that $\mathcal{R}$ returns $M_0^\prime, \dots, M_k^\prime$ and let all the subsequent probabilities being conditioned to such an event.
    Let $E$ and $\alpha_1^\prime, \dots, \alpha_k^\prime$ be the matrices and the scalars of steps~\ref{ite:R1} and~\ref{ite:R3} of $\mathcal{R}$, respectively.
    By construction, we have that $(M_1^\prime, \dots, M_k^\prime)$, $E$, and $(\alpha_1^\prime, \dots, \alpha_k^\prime)$ are independent and (conditionally) uniformly distributed in $\mathcal{C}_1$, $\mathcal{E}$, and $\mathbb{F}_q^k$, respectively.
    Moreover, again by construction, it follows that
    \begin{equation*}
        S^\prime := (M_0^\prime, \dots, M_k^\prime, E, \alpha_1^\prime, \dots, \alpha_k^\prime) \in \mathcal{S} .
    \end{equation*}
    Pick an arbitrary
    \begin{equation*}
        S^* := (M_0^*, \dots, M_k^*, E^*, \alpha_1^*, \dots, \alpha_k^*) \in \mathcal{S} .
    \end{equation*}
    By~\eqref{ite:n1}, we have that
    \begin{align*}
        \Pr[S^\prime = S^*] &= \Pr\!\left[(M_1^\prime, \dots, M_k^\prime) = (M_1^*, \dots, M_k^*)\right] \\
        &\phantom{mm}\cdot \Pr[E = E^*] \cdot \Pr\!\left[(\alpha_1^\prime, \dots, \alpha_k^\prime) = (\alpha_1^*, \dots, \alpha_k^*)\right] .
    \end{align*}
    Hence, we get that $S^\prime$ is uniformly distributed in $\mathcal{S}$.

    Let $M_0^{\circ}, \dots, M_k^{\circ}$, $K^{\circ}$, $E^{\circ\mathrm{R}}$,
    and $\alpha_1^{\circ}, \dots, \alpha_k^{\circ}$ be the matrices and the scalars generated by \textsf{KeyGen3}.
    Also, put $E^{\circ} := \big(E^{\circ\mathrm{R}} K^{\circ} \mid E^{\circ\mathrm{R}}\big)$.
    It follows easily that
    \begin{equation*}
        S^{\circ} := (M_0^{\circ}, \dots, M_k^{\circ}, E^{\circ}, \alpha_1^{\circ}, \dots, \alpha_k^{\circ}) \in \mathcal{S} .
    \end{equation*}
    Moreover, it follows easily that $M_0^{\circ\mathrm{R}}$, $(M_1^{\circ}, \dots, M_k^{\circ})$, and $K$ are independent and uniformly distributed in $\mathbb{F}_q^{m \times r, r}$, $\mathcal{C}_1$, and $\mathbb{F}_q^{r \times (n - r)}$, respectively.
    Therefore, by~\eqref{ite:n2}, we have that
    \begin{align*}
        \Pr[S^{\circ} = S^*] &= \Pr[M_0^{\circ \mathrm{R}} = M_0^{* \mathrm{R}}] \cdot \Pr\!\left[(M_1^\circ, \dots, M_k^\circ) = (M_1^*, \dots, M_k^*)\right] \\
        &\phantom{mm} \cdot \Pr[K^\circ = K^*] .
    \end{align*}
    Hence, we get that $S^\circ$ is uniformly distributed in $\mathcal{S}$.

    The proof is complete.

    \section*{Statements and Declarations}

    \subsection*{Competing Interests}
    The authors declare that they have no known competing financial interests or personal relationships that could have appeared to influence the work reported.

    \subsection*{Data availability statement}
    No new data were created or analysed in this study.
    Data sharing is not applicable to this article.


\end{document}